\documentclass[conference,letterpaper]{IEEEtran}

\usepackage[utf8]{inputenc} 
\usepackage[T1]{fontenc}
\usepackage{url}
\usepackage{ifthen}
\usepackage{cite}
\usepackage[cmex10]{amsmath} % Use the [cmex10] option to ensure complicance
\usepackage{amsthm,amsmath,amsfonts,braket,algorithm,graphicx, braket,balance}

\theoremstyle{definition} 
\newtheorem {theorem} {Theorem}

\newtheorem {lemma} {Lemma}

\usepackage{ifthen}
\newboolean{fullversionflag}
\setboolean{fullversionflag}{true}
\newcommand{\fullversion}[2]{\ifthenelse{\boolean{fullversionflag}}{{#1}}{{#2}}}

\newcommand{\num}{\#}

\newcommand{\kb}[1]{\left[#1\right]}

\newcommand{\al}{\mathcal{A}}

\newcommand{\trd}[1]{\left|\left| #1 \right| \right|}

\newcommand{\st}{\text{ } : \text{ }}

\newcommand{\Hmin}{H_\infty}

\newcommand{\leakEC}{\texttt{leak}_{EC}}

\newcommand{\up}[1]{^{({#1})}}

\floatname{algorithm}{Protocol}

\newcommand{\bit}{\mathcal{B}}

\newcommand{\ghz}{g}

\begin{document}
\title{Quantum Conference Key Agreement with Classical Advantage Distillation}

% %%% Single author, or several authors with same affiliation:
% \author{%
%  \IEEEauthorblockN{Author 1 and Author 2}
% \IEEEauthorblockA{Department of Statistics and Data Science\\
%                    University 1\\
 %                   City 1\\
  %                  Email: author1@university1.edu}% }

%%% Several authors with up to three affiliations:
\author{%
  \IEEEauthorblockN{Walter O. Krawec}
  \IEEEauthorblockA{School of Computing,
                    University of Connecticut\\
                    Storrs, CT, USA.\\
                    Email: walter.krawec@uconn.edu}
}

\maketitle

%%%%%%
%% Abstract: 
%% If your paper is eligible for the student paper award, please add
%% the comment "THIS PAPER IS ELIGIBLE FOR THE STUDENT PAPER
%% AWARD." as a first line in the abstract. 
%% For the final version of the accepted paper, please do not forget
%% to remove this comment!
%%

\begin{abstract}
In this work, we prove security of a quantum conference key agreement (QCKA) protocol augmented with a classical advantage distillation (CAD) protocol.  We derive a proof of security, in the finite key setting, that is able to bound the secure key rate for any general, coherent, attack.  We evaluate the performance of the system, showing our result can improve the noise tolerance of the protocol in some, but not all, scenarios.
\end{abstract}

\section{Introduction}%\label{sec:introduction:}

Quantum key distribution (QKD) allows two parties to establish a shared secret key, secure against computationally unbounded adversaries.  This is a very mature field at this point, with numerous, and fascinating, results both theoretically and experimentally.  The study of QKD security is also vital for other novel quantum cryptographic primitives beyond key distribution.  Moving beyond QKD, other related cryptographic tasks are possible, one of which is group-key distribution, also known as \emph{quantum conference key agreement} (QCKA).  QCKA protocols allow a group of users to establish a joint, shared, secret key.  While parallel instances of point to point QKD may be used in a trivial manner to achieve this, there are several advantages to QCKA protocols \cite{epping2017multi}.  For a general survey, the reader is referred to \cite{murta2020quantum}.

In this work, we take a BB84-style QCKA protocol, originally introduced in \cite{grasselli2018finite}, and consider how classical advantage distillation (CAD) \cite{maurer1993secret} can be applied to improve its performance in high noise scenarios.  CAD protocols work by dividing the raw key into blocks and attempting to establish a smaller, but more correlated, raw key from the original.  While we only consider block sizes of two in this paper, to our knowledge, this is the first proof of security for this QCKA protocol with any form of CAD.  Our proof bounds the quantum min entropy of the system directly, without having to rely on approximate bounds based on von Neumann entropy, which can add significant costs to the final key-rate bound.

We make several contributions in this work.  First, we propose the study of CAD for the QCKA protocol introduced in \cite{grasselli2018finite} and derive a new proof of security for this system in the finite key setting.  Our proof uses the quantum sampling framework of Bouman and Fehr \cite{bouman2010sampling} as a foundation.  Our proof does not require any assumptions on the attack (e.g., it works for general coherent attacks), though we do assume loss-less, single qubit, channels (we do not consider multi-photon sources).  Our proof only handles block sizes of two, leaving alternative block sizes as future work.  However, despite that, our proof methods may be broadly applicable to studying other QKD and QCKA protocols with CAD type methods, without having to rely on attack assumptions, or min entropy approximations based on bounds of von Neumann entropy (as is often considered in CAD-enabled QKD proofs).  Finally, we evaluate our key-rate bound and show when it is better, and when it is worse, than prior work.

%While our result is not always better, there are realistic cases where it does outperform prior work without CAD. We hope that this initial, preliminary, investigation might spur future research in this area, and our proof method may help with such an endeavor.
Quantum cryptography relies on quantum and classical methods, and the careful use of the latter can greatly benefit the protocol.  Despite this, however, proving security of such systems is often a challenge.  We hope our methods may be broadly applicable to other researchers investigating these areas, and that our evaluations show the potential importance of doing so to QCKA.

% subsection{Preliminaries}%\label{sec:introduction:}
\textbf{Preliminaries: }
Let $\al_d = \{0,1,\cdots, d-1\}$.  Given $q \in \al_d^N$ and $c \in \al_d$, let $\num_c(q)$ be the number of times $c$ appears in the word $q$.  We write $wt(q)$ to be the number of times a non-zero character appears in $q$ and $w(q)$ to be the relative Hamming weight of $q$, namely $w(q) = \frac{1}{N}wt(q)$.

Given a word $q\in\al_d^N$, we write $q_i$ to mean the $i$'th character of $q$ (with $1$ being the first character).  Given a subset $t \subset \{1, \cdots, N\}$ of size $|t| = m$, we write $q_t$ to mean the substring of $q$ indexed by $t$, namely $q_t = q_{t_1}q_{t_2}\cdots q_{t_m}$.  We write $q_{-t}$ to mean the substring of $q$ indexed by the complement of $t$.

We denote by $\bit_p = \{0,1\}^p$.  Given $x \in \bit_p^n$, we mean a $pn$-bit string which can be decomposed as $x_1\cdots x_n$, with each $x_i\in\{0,1\}^p$ indexed by $x_{i,1}\cdots x_{i,p}$.

Let $M = \{\ket{m_0}, \cdots, \ket{m_{d-1}}\}$ be some orthonormal basis.  Given $i\in\al_d$, we write $\ket{i}^{M}$ to mean the $i$'th basis vector in $M$, namely $\ket{i}^{M} = \ket{m_i}$.  Given $q \in \al_d^N$, we write $\ket{q}^{M}$ to mean $\ket{m_{q_1}}\ket{m_{q_2}}\cdots\ket{m_{q_N}}$.  If the basis is not specified, (i.e., $\ket{q}$), then we assume the standard computational basis.

One important basis we will work with is the GHZ basis.  Let $x \in \bit_p$ and $y\in\{0,1\}$.  Then the GHZ basis, for a $p+1$ qubit space, is defined to be those states:
$\ket{g^{p+1}(x;y)} = \frac{1}{\sqrt{2}}(\ket{0,x} + (-1)^y\ket{1, \bar{x}}),$
where $\bar{x}$ is the bitwise complement of $x$.  If the context is clear, we will forgo writing the ``$p+1$'' superscript.  Given $x\in\bit_p^n$ and $y\in\{0,1\}^n$, we will write $\ket{g(x;y)}$ to mean $\ket{g(x_1;y_1)}\otimes\cdots\otimes\ket{g(x_n;y_n)}$ with $x_i$ being the $i$'th block of $p$ bits in $x$ as discussed earlier in this section.

The following lemma will be important later which is easily proven:
\begin{lemma}\label{lemma:introduction:ghz-x}
Let $\ket{g^{p+1}(x;y)}$ be a GHZ state and assume a measurement is made of all qubits in the Hadamard basis resulting in outcome $q\in\bit_{p+1}$.  Then,  $y = q_1\oplus \cdots \oplus q_{p+1}$.
\end{lemma}

Let $\rho_{AE}$ be some quantum state acting on Hilbert space $\mathcal{H}_A\otimes\mathcal{H}_E$.  Then we write $\rho_A$ to mean the result of tracing out the $E$ portion of $\rho_{AE}$, namely $\rho_A  = tr_E\rho_{AE}$.  This notation extends to multiple systems.  Given a vector $\ket{z}$, we write $\kb{z}$ to mean $\kb{z} = \ket{z}\bra{z}$.  Later, we will also use the projection operator denoted $P(\ket{z}) = \kb{z}$.

We use $H(A)$ to mean the classical Shannon entropy of random variable $A$ and $h(x)$ to mean the binary Shannon entropy of $x\in[0,1]$.  Given a quantum state $\rho_{AE}$ we write $\Hmin(A|E)_\rho$ to mean the conditional quantum min entropy defined in \cite{renner2008security}.  The \emph{smooth quantum min entropy} is denoted $\Hmin^\epsilon(A|E)$ and is defined by \cite{renner2008security} $\Hmin^\epsilon(A|E) = \sup_\sigma\Hmin(A|E)_\sigma,$
where the supremum, above, is over all density operators $\sigma_{AE}$ satisfying $\trd{\rho_{AE}-\sigma_{AE}} \le \epsilon$.  Here, $\trd{A}$ is the trace distance of operator $A$.

Quantum min entropy is a vital resource in cryptography as it measures how many secret uniform random bits may be extracted from a given classical-quantum state $\rho_{AE}$ through privacy amplification.
%% Namely, if Alice were to take the $A$ register and run a privacy amplification protocol on it (which, itself, consists of choosing a random two-universal hash function and running $A$ through it \cite{renner2008security}), then the distance in the output from a secret key is bounded by a function of the min entropy.  In particular, assume the $A$ register is an $N$-bit (classical) string (with the $E$ system a quantum state potentially entangled with it).
Namely, if one chooses a random two-universal $f:\{0,1\}^N \rightarrow \{0,1\}^\ell$, with $\ell \le N$, it holds \cite{renner2008security}:
\begin{equation}\label{eq:introduction:pa}
\trd{\rho_{F(A), EF} - 2^{-\ell}I\otimes\rho_{EF}} \le \sqrt{2^{l - \Hmin^\epsilon(A|E)_\rho}} + 2\epsilon,
\end{equation}
where, above, $F$ represents the register storing the choice of hash function (which we can assume Eve has access to).  %Essentially, above, the larger the min entropy, the closer the ``real'' system is ($\rho_{F(A),EF}$) to the ``ideal'' system (namely, $2^{-\ell}I\otimes\rho_{EF}$ which denotes a uniform random key chosen independenlty of Eve's system).

We will make use of several important properties of quantum min entropy.  First, given a state classical in $Z$, namely $\rho_{AEZ} = \sum_zp(z)\rho_{AE}\up{z}$, then:
\begin{equation}\label{eq:introduction:min-mixed-cl}
\Hmin(A|E)_\rho \ge \Hmin(A|EZ)_\rho \ge \min_z\Hmin(A|E,Z=z)_\rho,
\end{equation}
where $\Hmin(A|E,Z=z)_\rho = \Hmin(A|E)_{\rho\up{z}}$.

% Another nice property of quantum min entropy is that we can compute the entropy of a measurement performed on a superposition state $\ket{\psi}_{AE}$ by computing the min entropy of a suitable mixed state.
The following two lemmas will be important later:
\begin{lemma}\label{lemma:introduction:ent-super}
  (From \cite{bouman2010sampling} based on a proof in \cite{renner2008security}):  Let $\ket{\psi}_{AE} = \sum_{q\in J}\alpha_q\ket{q}^M\ket{E_q}$, with $J \subset \al_d^n$ be a quantum state with respect to some $d$-dimensional basis $M$.  Consider the mixed state $\sigma = \sum_{q\in J}|\alpha_q|^2\kb{q}^M\otimes\kb{E_q}$.  Assume a measurement is performed of the $A$ system using some other basis $N$ resulting in outcome $\rho_{A_NE}$ in the first case (the pure state) or $\sigma_{A_NE}$ in the second (mixed).  Then it holds that:
$    \Hmin(A_N|E)_\rho \ge \Hmin(A_N|E)_\sigma - \log_2|J|.$
%  Namely, the entropy of the pure state can be bounded by the entropy in the mixed state (which is often easy to compute), up to the size of the superposition set $J$.
\end{lemma}

%Another useful lemma allows us to bound the quantum min entropy in a state conditioning on a particular measurement outcome, based on the min entropy of a state that is ``close'' to the given one.
\begin{lemma}\label{lemma:introduction:entropy-mixed-cptp}
  (From \cite{krawec2022security}): Let $\rho$ and $\sigma$ be two quantum states with $\frac{1}{2}\trd{\rho-\sigma}\le \epsilon$.  Let $\mathcal{F}$ be a CPTP map such that $\mathcal{F}(\rho) = \sum_xp_x\kb{x}_X\otimes\rho_{AE}\up{x}$ and $\mathcal{F}(\sigma) = \sum_xq_x\kb{x}_X\otimes\sigma_{AE}\up{x}$.  Specifically, $\mathcal{F}$ is a map that measures part of the system leaving a conditional state based on that observation (the conditional state may also have been measured or acted on in some way by $\mathcal{F}$).  Then it holds that:
  \begin{equation}%\label{eq:introduction:}
    Pr\left(\Hmin^{4\epsilon+2\epsilon^{1/3}}(A|E)_{\rho\up{x}} \ge \Hmin(A|E)_{\sigma\up{x}}\right) \ge 1-2\epsilon^{1/3}.
  \end{equation}
  where, above, the probability is over the random system $X$.
\end{lemma}

Finally, we can prove the following lemma for GHZ states.  This follows immediately from Lemma \ref{lemma:introduction:ent-super} and basic identities of GHZ states:
\begin{lemma}\label{lemma:introduction:ghz-entropy}
  Let $\ket{\psi} = \sum_{x\in\bit_p^n}\sum_{y\in J}\ket{g^{p+1}(x;y)}\ket{E_{x,y}}$ be some quantum state (the $\ket{E_{x,y}}$ states, here, are sub-normalized) where $J \subset \{0,1\}^n$.  Assume a measurement is made in the $Z$ basis of the first qubit of each of the $n$ GHZ states, while the remaining $p$ qubits are discarded. This process results in a density operator $\rho_{ZE}$.  Then it holds that:
%  \begin{equation}%\label{eq:introduction:}
$    \Hmin(Z|E)_\rho \ge n - \log_2|J|.$
%  \end{equation}
\end{lemma}
\begin{proof}%\label{pf:introduction:}
  Consider, first, a single GHZ state $\ket{g(x_i;y_i)}$.  Writing this state in the Hadamard basis produces:
%  \begin{align}%
%    \frac{1}{\sqrt{2^{p+2}}}\sum_{c_0,\cdots, c_p\in\{0,1\}}\left((-1)^{c\cdot x} + (-1)^{y+c_0+c\cdot\bar{x}}\right)\ket{c_0,\cdots, c_p}^X.
     %    \end{align}
%  \begin{equation}%\label{eq:introduction:}
$    \frac{1}{\sqrt{2^{p+2}}}\sum_{c_0,\cdots, c_p}(-1)^{c\cdot x}\left(1 + (-1)^{y+c_0+c_1+\cdots c_p}\right)\ket{c_0,\cdots, c_p}^X.$
%  \end{equation}
  where $c\cdot x = c_1x_1 + c_2x_2 + \cdots c_px_p$ (with the addition done modulo two).%  Since $c_ix_i \oplus c_i\bar{x_i} = c_i$, this simplifies to:
  The coefficient of $\ket{c_0,\cdots, c_p}$ is non-zero only when $c_0 = y\oplus c_1\oplus\cdots\oplus c_p$.  Thus, we may write the above as $\ket{g(x_i;y_i)}=$
%  \begin{equation}%\label{eq:introduction:}
$    \frac{1}{\sqrt{2^p}}\sum_{c_1,\cdots,c_p}(-1)^{c\cdot x}\ket{y\oplus c_1\oplus\cdots\oplus c_p, c_1,\cdots,c_p}^X.$
%  \end{equation}
  Now, returning to the original state $\ket{\psi}$, we change basis of all $n$ GHZ states and then trace out those systems which will not be measured.  The resulting state can be written as a mixture, over $c \in \bit_p^n$ of states of the form $\sum_{y\in J}\ket{y\oplus f(c)}\ket{F_{y}\up{c}},$
%  \begin{equation}%\label{eq:introduction:}
%    \rho_{XE} = \sum_{c\in\bit_p^n}p(c)P\left(\sum_{y\in J}\ket{y\oplus f(c)}\ket{F_{y}\up{c}}\right),
%  \end{equation}
  where $f(c)$ is a deterministic function of $c$ and $\ket{F_y\up{c}}$ is a function of the $\ket{E}$ vectors, along with $c$. The result follows from Lemma \ref{lemma:introduction:ent-super} and Equation \ref{eq:introduction:min-mixed-cl}.
\end{proof}

%$ $\newline
%\textbf{Quantum Sampling: }
Our proof will make use of the sampling framework of Bouman and Fehr introduced in \cite{bouman2010sampling}, along with proof techniques we developed for sampling-based entropic uncertainty relations \cite{krawec2019quantum}.  We only briefly cover the sampling framework here, leaving more details to \cite{bouman2010sampling}.

Let $\delta > 0 $ and $N\in\mathbb{N}$ be given, and let $t \subset \{1, \cdots, N\}$.  We define a set of \emph{good words} $\mathcal{G}^t_\delta$ (with respect to subset $t$ and given $\delta$), to be a subset of $\al_d^N$ (or some other alphabet) of words that obey some structure.% For instance:
%\begin{equation}\label{eq:introduction:good-HW}
%\mathcal{G}^t_\delta = \{q \in \{0,1\}^N \st |w(q_t) - w(q_{-t})|\le\delta\}.
%\end{equation}
%Other choices are, of course, possible.
In general, the set $\mathcal{G}^t_\delta$ is induced by fixing a classical sampling strategy on words in $\al_d^N$.  Fix $\delta$, along with a probability distribution $P_T$ over subsets $t \subset\{1, \cdots, N\}$. Then, we can define the so-called \emph{error probability} of the strategy:
%\begin{equation}%\label{eq:introduction:}
$\epsilon^{cl}_\delta = \max_qPr\left(q\not\in\mathcal{G}^t_\delta\right),$
%\end{equation}
where the probability, above, is over the choice of subsets $t$, according to the given distribution $P_T$.  %Essentially, if $\epsilon^{cl}_\delta$ is small, it can be guaranteed that, on average over the subset choice, any word will obey the given sampling strategy (e.g., in the case of Equation \ref{eq:introduction:good-HW}, the Hamming weight of $w(q_{t})$ will be close to $w(q_{-t})$ with high probability over the subset choice).

Now, fix a $d$-dimensional orthonormal basis $\mathcal{M}$ and consider the following subspace:
%\begin{equation}%\label{eq:introduction:}
$\mathcal{G}^t_{\delta}(\mathcal{M}) = \text{span}\left(\ket{q}^{\mathcal{M}} \st q \in \mathcal{G}^t_\delta\right)\otimes\mathcal{H}_E,$
%\end{equation}
where $\mathcal{H}_E$ is some arbitrary ancilla system (perhaps controlled by an adversary system).  Then, the following theorem was proven in \cite{bouman2010sampling}:
\begin{theorem}\label{thm:introduction:sample}
  (From results in \cite{bouman2010sampling}): Let $\ket{\psi}_{AE}$ be a quantum state, where the $A$ register is of dimension $d^N$ and let $\mathcal{M}$ be a $d$-dimensional orthonormal basis.  Let $\delta > 0$ and fix a family of good words $\mathcal{G}^t_\delta$, along with a probability distribution over subsets $P_T$.  Then, there exists a collection of \emph{good states}, $\ket{\phi^t}_{AE}$ such that for every $t$, it holds $\ket{\phi^t}\in \mathcal{G}^t_\delta(\mathcal{M})$ and:
%  \begin{equation}%\label{eq:introduction:}
$    \frac{1}{2}\trd{\sum_tP_T(t)\kb{t}(\kb{\psi}_{AE} - \kb{\phi^t}_{AE})} \le \sqrt{\epsilon_\delta^{cl}}.$
%  \end{equation}
  
\end{theorem}
\begin{proof}%\label{pf:introduction:}
See \cite{bouman2010sampling}; note that they actually proved a stronger result, however we word the theorem in this way as it becomes immediately applicable to our work.  For a proof that this wording follows immediately from their results, see \cite{yao2022quantum}.
\end{proof}

Before leaving this section, we conclude with the following:
\begin{lemma}\label{lemma:introduction:basic-sample}
  (From \cite{bouman2010sampling}): Let $\mathcal{G}^t_\delta = \{q\in\{0,1\}^N \st |w(q_t) - w(q_{-t})|\le\delta\}$.  Also, let $P_T$ be a distribution on subsets which chooses subsets $t$ of fixed size $m$ uniformly at random for some given $m < N/2$.  Then $\epsilon^{cl}_\delta = \max_{q\in\{0,1\}^N}Pr\left(q\not\in\mathcal{G}^t_\delta\right) \le \epsilon_0 = 2\exp\left(\frac{-\delta^2mN}{N+2}\right).$
%  \begin{equation}%\label{eq:introduction:}

%  \end{equation}
  
\end{lemma}

%%% Local Variables:
%%% mode: latex
%%% TeX-master: "main"
%%% End:

\section{The Protocol}\label{sec:protocol:protocol}

The protocol we consider in this work is the GHZ-based BB84 protocol, first introduced in \cite{grasselli2018finite}, with the addition of a two-block CAD protocol similar to \cite{maurer1993secret}, however extended, in the natural way, to work with multiple parties.  To save space, we only write out the entanglement based version of the protocol.

Let $p+1$ be the number of users; we label the users Alice (who will be the ``leader'' in a way) and the other users Bob$^1$ through Bob$^p$.  Let $2N$ be the total number of rounds (or signals) used, and let $2m < N$ (i.e., $m < N/2$) be the size of the sample used for error testing.  The quantum communication stage operates as follows:

1. Eve creates a quantum state $\ket{\psi}_{AB^1\cdots B^pE}$, where each $A$ and $B^i$ system consists of $2N$ qubits each; the $E$ system is arbitrary.  If the source is honest and there is no noise, the state should consist of $2N$ independent copies of $\ket{\ghz^{p+1}(0;0)}$. The $A$ register is sent to Alice while the $B^i$ register is sent to Bob$^i$.  Alice chooses a permutation and announces it to each Bob; all parties permute their $2N$ systems according to this permutation.  They then divide their $2N$ qubits into a Left half (the first $N$) and a Right half (the last $N$ qubits).

2. Alice chooses a subset $t \subset \{1, \cdots, N\}$ with $|t| = m$, uniformly at random and sends this subset to all other users.  For every qubit $j \in t$, parties will measure their $j$'th Left qubit and their $j$'th right qubit in the $X$ basis.  This results in measurement outcomes $l^0, r^0, l^1,r^1 \cdots, l^p,r^p \in \{0,1\}^m$, where $l^0,r^0$ is Alice's outcome for her left qubit measurement (respectively Right); and similarly $l^j,r^j$ for $j>0$ will be Bob$^j$'s outcome.  Let $q^j = l^j\oplus r^j$.  Parties announce to one another the XOR of their measurement outcomes (namely they announce $q^j$) allowing them to compute $Q_X = w(q^0\oplus\cdots \oplus q^p)$.  Note that, ideally, it should hold that $Q_X = 0$ (see also Lemma \ref{lemma:introduction:ghz-x}).  Note, also, that all classical communication done is over the authenticated classical channel.

3. The remaining $2n = 2N-2m$ signals are measured in the $Z$ basis by all parties.

Ordinarily, the QCKA protocol of \cite{grasselli2018finite} now runs error correction and privacy amplification.  However, before this, we propose that the following CAD process is run:

1. Let $k^0$ be Alice's raw key (the resulting $Z$ basis measurement outcomes in the last step of the above process); similarly let $k^j$ be Bob$^j$'s.  Each raw key may be divided into left and right portions (as with the testing phase above), denoted $k^j=k^j_{L}k^j_{R}$.  For each $i=1, 1, \cdots, N-m$, Alice broadcasts the parity of her left and right raw key bits; namely she broadcasts $p_i = k^0_{L,i}\oplus k^0_{R,i}$.  This is done by sending the resulting parity string to each Bob over the authenticated channel.  Each Bob$^j$ will similarly broadcast the parity of his left and right raw key bits.

2. For each two-bit block $i=1, \cdots, N-m$, if one or more Bob's report a parity outcome that is different from Alice's, that block is discarded by all parties.  Otherwise, if all Bob's have the same parity for the $i$'th block, they will keep the first bit of that block (the Left bit) as their new raw key, discarding the second bit (Right bit).  We use $n_a$ to be the number of accepted blocks (which is also the size of the new raw key).

After the above, error correction and privacy amplification are run as normal.
%The above process results in a new raw key that is smaller, but hopefully more correlated with the leader Alice. After the above CAD process is run, error correction and privacy amplification are run as normal.  %Note that, as proven in \cite{grasselli2018finite}, error correction need only leak an amount of information proportional to the Alice-Bob$^j$ pair with the highest raw key error rate.

%%% Local Variables:
%%% mode: latex
%%% TeX-master: "main"
%%% End:

\section{Security}%\label{sec:security:}

We now prove security of the described QCKA protocol with CAD.  We comment that most, if not all, proofs of security involving some CAD process (even for basic two-party QKD) involve bounding the von Neumann entropy assuming i.i.d. attacks, and then using various results which promote the analysis to coherent attacks in the finite key setting.  However those results tend to have a significant error term, especially for small signal sizes, which must be deducted from the final key rate.  Our proof, instead, bounds directly the quantum min entropy for arbitrary coherent attacks, using as a foundation, the sampling framework of Bouman and Fehr discussed earlier.

First, we must define a suitable sampling set.  Let $N$ be fixed and let $t \subset \{1, \cdots, N\}$ with $|t| = m$ for given $m < N/2$.  Let $p$ be the number of Bob's in the protocol.  We define the set of good words for our strategy to be:
\begin{align}\label{eq:security:good-words}
  \mathcal{G}^t_\delta = \{&(x,y,z,w) \in \bit_p^N\times\{0,1\}^N\times\bit_p^N\times\{0,1\}^N\notag\\
  &\st |w(y_t\oplus w_t) - w(y_{-t}\oplus w_{-t})|\le\delta\}.
\end{align}
Note that the good words do not depend on the ``$x$'' and ``$z$'' portion of the tuples, only the $y$ and $w$ values.  It is not difficult to show using Lemma \ref{lemma:introduction:basic-sample}, that for any $q = (x,y,z,w)$:
\begin{equation}\label{eq:security:fail-pr}
  \epsilon^{cl}_\delta = \max_qPr(q \not \in \mathcal{G}^t_\delta) \le \epsilon_0.
\end{equation}
where $\epsilon_0$ was defined in Lemma \ref{lemma:introduction:basic-sample}.  To prove the above, note that for any $q$, an equivalent bit string $\tilde{q}$ of size $N$, may be constructed such that $q \not \in \mathcal{G}^t_\delta$ will imply $\tilde{q}$ is not in the set defined and analyzed in Lemma \ref{lemma:introduction:basic-sample}, for all subsets $t$.  Namely the sampling strategy we care about can be reduced to the one analyzed in Lemma \ref{lemma:introduction:basic-sample}.

We are now in a position to state and prove our main result:
\begin{theorem}%\label{thm:security:}
  Let $\epsilon > 0$ and let $\ket{\psi}_{AB^1\cdots B^pE}$ be a quantum state distributed by Eve where the $A$ and $B^j$ registers consist of $2N$ qubits each.  Let $\rho_{AE}\up{t,Q_X,n_a}$ be the final state after running the QCKA protocol with CAD conditioned on a particular subset $t$ being chosen, along with a particular observation $Q_X$ and number of accepted (not discarded) blocks $n_a$.  Here, the $E$ register contains all public information (e.g., parity announcements).  Then, it holds that:
  \begin{equation}%\label{eq:security:}
    \Hmin^{\epsilon'}(A|E)_{\rho\up{t,QX,n_a}} \ge n_a\left(1 - h\left[\frac{n}{n_a}(Q_X+\delta)\right]\right)
  \end{equation}
  except for probability at most $\epsilon_{fail}$, where the probability is over all subsets $t$, observations $Q_X$, and values of $n_a$, and where $\epsilon_{fail} = 2\epsilon^{1/3}$.  Above, $n = N-m$, the smoothening parameter is $\epsilon' = 4\epsilon+2\epsilon^{1/3}$, and
%  \begin{equation}%\label{eq:security:}
$    \delta = \sqrt{\frac{(N+2)\ln(2/\epsilon^2)}{mN}}.$
%  \end{equation}
%  Note that the total number of GHZ states distributed in the protocol is $2N$ of which $2m$ were used for the sampling process.
\end{theorem}
\begin{proof}%\label{pf:security:}
   Fix a permutation $\pi$ used by Alice and the Bob's in step one of the protocol and let $\ket{\widetilde{\psi}}_{ABE}$ be the resulting state after the parties permute their systems.  Consider the set of good words defined in Equation \ref{eq:security:good-words}; for any $t$, consider the quantum subspace spanned by $\mathcal{G}^t_\delta(GHZ) = \{\ket{g(x;y)}\ket{g(z;w)} \st (x,y,z,w) \in \mathcal{G}^t_\delta\}\otimes\mathcal{H}_E$.  By Theorem \ref{thm:introduction:sample}, there exists ideal states $\ket{\phi^t}$ such that $\ket{\phi^t} \in \mathcal{G}^t_\delta(GHZ)$ and:
\begin{equation}\label{eq:security:trd-ideal}
\frac{1}{2}\trd{\sum_tP_T(t)\kb{t}(\kb{\widetilde{\psi}} - \kb{\phi^t})}\le \epsilon,
\end{equation}
where the last inequality follows from Equation \ref{eq:security:fail-pr} and our choice of $\delta$.  Let's consider tracing the protocol over the ideal state $\sigma_{TABE} = \sum_tP_T(t)\kb{t}_T\kb{\phi^t}_{ABE}$.  Choosing a subset $t$ is equivalent to measuring the $T$ register, causing the state to collapse to a particular good state $\ket{\phi^t}_{ABE}$.  From Lemma \ref{lemma:introduction:ghz-x}, measuring in the $X$ basis, those pairs of GHZ states as discussed above, resulting in outcome $Q_X$, causes the ideal state to collapse to:
\begin{equation}\label{eq:security:ideal-pre-measure}
\sigma\up{t,Q_X} = \sum_{\substack{(b,d) \in \mathcal{I}_{Q_X}\\a,c\in\bit_p^m}}p_{a,b,c,d}\kb{\mu(a,b,c,d)}
\end{equation}
where:
%\begin{equation}%\label{eq:security:}
$\mathcal{I}_{Q_X} = \{(b,d) \in \{0,1\}^m \st w(b\oplus d) = Q_X\}$
%\end{equation}
and $\ket{\mu(a,b,c,d)} = $
\begin{equation}\label{eq:security:mu}
\sum_{\substack{(y,w)\in J(Q_X,b,d)\\x,z\in\bit_p^n}}\ket{g(x;y)}\ket{g(z;w)}\ket{E_{x,y,z,w}^{t,Q_X,a,b,c,d}},
\end{equation}
with $J(Q_X:b,d)=$
%\begin{align*}%\label{eq:security:}
  $\{(y,w) \in \{0,1\}^n \st |w(y\oplus w) - w(b\oplus d)| \le \delta\}  =\{(y,w) \in \{0,1\}^n \st |w(y\oplus w) - Q_X| \le \delta\}.$
%\end{align*}
The above follows from Lemma \ref{lemma:introduction:ghz-x} (for the definition of $\mathcal{I}_{Q_X}$) and from the fact that the state measured lives in the subspace of good words (for the definition of $J(Q_X;b,d)$).  Note that the latter set does not depend on $b,d$ and so we will forgo writing it in the future and simply write $J(Q_X)$.  We will call the GHZ states indexed by $x$ and $y$ the Left states and the other GHZ states the Right states.  

Now, normally, Alice and the Bob's will measure their remaining qubits in the $Z$ basis and run the CAD protocol.  However, instead, we will use  a delayed measurement technique, modeling the CAD process through unitary operations, followed by a measurement after the fact.  This will be mathematically equivalent to the actual protocol.  First, Alice, will create a new $n$ qubit ancilla (called the \emph{parity ancilla}) initialized in the $\ket{0\cdots 0}$ state which will store her parity announcements.  Then, for $j=1, \cdots, n$, Alice will apply a CNOT with the control on her $j$'th Left qubit and the target on the $j$'th parity ancilla, followed by a second CNOT controlled on her $j$'th Right qubit and the target being, again, the $j$'th parity ancilla.  This, effectively, causes the parity ancilla to be the XOR of Alice's two qubits.  Measuring this ancilla later is equivalent to measuring the Left and Right qubits first, then XOR'ing the results classically.  Let's consider a single round $i$, and single state $\ket{g(x_i,y_i)}_L\ket{g(z_i;w_i)}_R$ (where, keep in mind, $x_i$ and $z_i$ are $p$-bit strings).  This operation will evolve this state, and the $i$'th qubit of the new parity ancilla, to:
\begin{align*}%
  &\ket{0}_P\ket{g(x_i,y_i)}_L\ket{g(z_i;w_i)}_R\\
  &\mapsto \frac{1}{2}\ket{0}_P(\ket{0,x_i,0,z_i} + (-1)^{y+w}\ket{1,\bar{x}_i,1,\bar{z}_i})\\
  &+ (-1)^w\frac{1}{2}\ket{1}_P(\ket{0,x_i,1,\bar{z}_i} + (-1)^{y+w}\ket{1,\bar{x}_i,0,z_i})
\end{align*}

At this point, each Bob will similarly perform this operation on his left and right qubits to store the parity in new ancilla registers for each Bob.  This will map the above state to:
%\begin{align*}%
%  &\frac{1}{2}\ket{0}_P\ket{x_i\oplus z_i}(\ket{0,x_i,0,z_i} + (-1)^{y+w}\ket{1,\bar{x}_i,1,\bar{z}_i}) +\\
%  &(-1)^w\frac{1}{2}\ket{1}_P\ket{x_i\oplus \bar{z_i}}(\ket{0,x_i,1,\bar{z}_i} + (-1)^{y+w}\ket{1,\bar{x}_i,0,z_i})
%  &=\frac{1}{\sqrt{2}}\sum_{r}\ket{r}_P(-1)^{w_i\cdot r}\ket{x_i\oplus r^p\oplus z_i}\ket{g(x_ir(r^p\oplus z_i); y_i\oplus w_i)}
%\end{align*}
 $\frac{1}{\sqrt{2}}\sum_{r}\ket{r}_P(-1)^{w_i\cdot r}\ket{x_i\oplus r^p\oplus z_i}\ket{g(x_ir(r^p\oplus z_i); y_i\oplus w_i)}$
where, now the GHZ state above contains $2p+2$ qubits and where $r^p = r\cdots r$ ($p$-times).
Note that, parties will accept this round only if the parity announcements of each Bob matches the parity announcement of Alice.  This is only true if $x_i=z_i$ (i.e., all $p$-bits in $x_i$ match $z_i$).  Of course, the above is applied to all $n$ of the remaining signal rounds on the superposition state in Equation \ref{eq:security:ideal-pre-measure}.  Let's just focus on a single $\ket{\mu(a,b,c,d)}$ for some $a,c$ and some $b,d\in \mathcal{I}_{Q_X}$.  If we can find a min entropy bound in this case, we can use Equation \ref{eq:introduction:min-mixed-cl} to compute the min entropy of the mixed state $\sigma\up{t,Q_X}$.  Tracing the above delayed measurement process yields a new state which may be written in the form:
\begin{align}
  &\frac{1}{\sqrt{2^n}}\sum_{r\in\{0,1\}^n}\ket{r}_P\sum_{x,z\in\bit_p^n}\ket{B(x,z,r)}\otimes\label{eq:security:state-before-measure}\\
  &\sum_{y,w\in J(Q_X)}(-1)^{w\cdot r}\ket{g(F(x,z,r); y\oplus w)}\ket{E_{x,y,z,w}\up{t,Q_X,a,b,c,d}}.\notag
\end{align}
where:
%\begin{equation}%\label{eq:security:}
$B(x,z,r) = x_1\oplus r_1^p\oplus z_1,  \cdots , x_n\oplus r_n^p\oplus z_n,$
%\end{equation}
and
%\begin{equation}%\label{eq:security:}
$F(x,z,r) = x_1r_1(r_1^p\oplus z_1), \cdots, x_nr_n(r_n^p\oplus z_n).$
%\end{equation}
At this point, parties will reject any round $i$ where  $x_i \ne z_i$.  For this, they measure their parity ancillas, and, based on the result, set a string $acc\in\{0,1\}^n$ where if $acc_i = 1$, it holds that $x_i = z_i$.  Following this measurement, parties discard (trace out) all qubits involved in round $i$ where $acc_i = 0$, leaving only those qubits in rounds where $acc_i=1$ (though, importantly for our proof, still unmeasured).  Conditioning on a particular $acc$ string, let $n_a = \num_1(acc)$ (namely, the number of accepted rounds) and $n_r = \num_0(acc)$.  Tracing the execution of this measurement and the subsequent tracing out of registers, we can write the state in the form: $\sum_{r,x,z}p(r,x,z) \kb{r, B(x,z,r)} \otimes$% it is not difficult to show that the state above (Equation \ref{eq:security:state-before-measure}) can now be written in the form (again, conditioning on this $acc$ string):
\begin{align*}%\label{eq:security:}
%  &\sum_{r,x,z}p(r,x,z) \kb{r, B(x,z,r)} \otimes\\
  &\sum_{y_R,w_R\in R}P\left(\sum_{y_A,w_A \in A}\ket{g(F(x,z,r); y_A\oplus w_A)}\ket{F_{x,z,r,y,w}}\right).
\end{align*}
where, above, we have:
%\begin{equation}%\label{eq:security:}
$R = \{y,w \in \{0,1\}^{n_R} \st wt(y\oplus w) \le n(Q_X+\delta),$
%\end{equation}
and:
%\begin{equation}%\label{eq:security:}
$A = \{y,w \in \{0,1\}^{n_a}\st wt(y\oplus w) + wt(y_R\oplus w_R) \le n(Q_X+\delta)\}.$
%\end{equation}
Above, the states $\ket{F_{x,z,r,y,w}}$ may be readily derived from Equation \ref{eq:security:state-before-measure} and are functions of Eve's ancilla.  Their exact derivation, however, is not important.
Also, note that, above, the sum of $x$ and $y$, is over all $x_i=z_i$ when $acc_i=1$ and all $x_i\ne z_i$ when $acc_i = 0$. The value of $p(r,x,z)$ can be derived, however is not important for the remainder of the proof as min entropy will always assume the worst case anyway.  Finally, at this point, Alice will measure her first qubit in the Left system for each round that was accepted to produce her raw key.  Let $A_Z$ be the random variable storing this result.  From the above equation, we can conclude, using Lemma \ref{lemma:introduction:ghz-entropy}, that the min entropy of this measurement, for this particular outcome $acc$ is bounded by $n_a - \log_2|\widetilde{A}|$ (where $\widetilde{A} = \{z\in\{0,1\}^{n_a}\st wt(z) + wt(y_R\oplus w_R) \le n(Q_X+\delta)\}$.  Of course, the size of $\widetilde{A}$ depends on $y_R, w_R$, however we may assume the worst case that $wt(y_R\oplus w_R) = 0$ and so, using the well known bound on the volume of a Hamming ball, we can conclude $\log_2|\widetilde{A}| \le n_ah\left(\frac{n}{n_a}(Q_X+\delta)\right)$.

Of course, the above is the min entropy for a particular $\ket{\mu(a,b,c,d)}$, however it is clear that our bound was independent of the particular choice of $a,b,c,d$.  Thus, we conclude, that, conditioning on a particular $acc$ string, we have:
%\begin{equation*}%\label{eq:security:}
$\Hmin(A_Z|E)_{\sigma\up{t,Q_X,n_a}} \ge n_a(1-h(\frac{n}{n_a}(Q_X+\delta))).$
%\end{equation*}
Lemma \ref{lemma:introduction:entropy-mixed-cptp}, combined with Equation \ref{eq:security:trd-ideal}, completes the proof for this particular initial permutation (taking the random variable $X$ to be the choice of $t$, and the observation of $Q_X$ and $acc$).  However, since the above did not depend on the permutation (though it will affect the probability of observing a particular $Q_X$ for instance), the result follows for any permutation.  %(In particular, this is a worst-case result; a more fine-grained analysis of the permutation may lead to better results as Eve may be more uncertain in some cases potentially - we leave that as future work.)
\end{proof}

From this, a key-size bound may be found using Equation \ref{eq:introduction:pa}.  If we set $\ell$ to:
\begin{equation}%\label{eq:security:}
\ell = n_a\left(1 - h\left[\frac{n}{n_a}(Q_x+\delta)\right]\right) - \leakEC - 2\log_2\frac{1}{\epsilon},
\end{equation}
then we will have a secret key that is $\epsilon_{PA}$-secure, with $\epsilon_{PA} = 9\epsilon + 2\epsilon^{1/3}$, except with probability $\epsilon_{fail} = 2\epsilon^{1/3}$.  Note that the key-rate, then, will be $\ell/(2N)$ since $2N$ signals are required to establish the above key.

%%% Local Variables:
%%% mode: latex
%%% TeX-master: "main"
%%% End:

% \section{Evaluation}%\label{sec:evaluation:}
\textbf{Evaluation: }
To evaluate our key-rate bound, we will assume a channel that acts independently for all rounds, causing an $X$ basis error in a single round with probability $Q$.  From this, we have $Q_X = 2Q(1-Q)$.  This is because $Q_X$ is measuring the relative number of single errors in two-bit blocks.  We will assume the $Z$ basis error between Alice and Bob$^j$ is $QZ_j$ which may be different from $Q$ and even may be different for the different Bobs.   We can compute the expected value of $n_a$ to be:
$\frac{n_a}{n} = p_a = (QZ_1^2 + (1-QZ_1)^2)\times\cdots \times(QZ_p^2 + (1-QZ_p)^2).$
% The above is due to the fact that, to accept a particular block of two bits, all Bob's must have the same parity as Alice; this occurs if both of Bob$^j$'s outcomes are identical or both are flipped.  Note that our proof can handle asymmetry in the Bob error rates.
In our evaluations, we set $\epsilon = 10^{-36}$ (so that $\epsilon_{PA}$ and $\epsilon_{fail}$ are on the order of $10^{-12}$).  We optimize over $m$ in our evaluations.   Note that the expected raw key error between Alice and Bob$^j$ is $QZ_j^2/p_a$.  Using results from \cite{grasselli2018finite} we can bound  $\leakEC = n_a\max_jh(QZ_j^2/p_a) +\log\frac{2p}{\epsilon}$.

We compare with finite key-results of this protocol without CAD from \cite{grasselli2018finite}.  These results and comparisons are shown in Figures \ref{fig:evaluation:fig1} and \ref{fig:evaluation:fig2}.  Note that the ``Two Party'' results are evaluating standard BB84 with CAD using our result.  We note that our result performs worse than prior work without CAD whenever the noise is symmetric (except in the two-party case).  However, when the noise is asymmetric (with some Bob's having less noise than others), our result shows higher key-rates are possible.
%The reason for the poor performance in the symmetric case can be due to this particular CAD protocol, or our proof method.  Future work should investigate this to see if CAD can help under symmetric noise.
Note that asymmetries in the $Z$ basis noise can occur in large-scale networked implementations of this QCKA protocol \cite{oslovich2024efficient}, and so our results may be applicable in those scenarios.

\begin{figure}
  \centering
  \includegraphics[width=.47\linewidth]{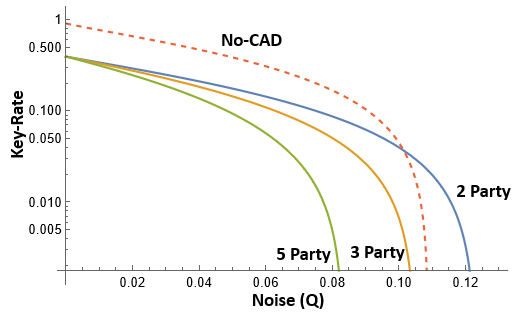}
  \includegraphics[width=.47\linewidth]{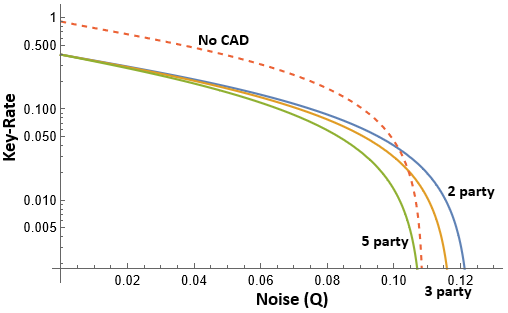}
  \caption{Showing the key-rate of our proposed CAD protocol (solid lines) and comparing with results in \cite{grasselli2018finite} for $10^7$ signals. Left: Here the noise is symmetric in that $QZ^j = Q$ for all $j$.  We note that our CAD key-rate is lower than without CAD whenever there are three or more parties.  Right: Here, the noise is asymmetric. In particular, $QZ^1 = Q$ while all other Bob's $j>1$ are $QZ^j = \frac{1}{4}Q$.  Here we see CAD can outperform standard QCKA protocol rates.  Similar results appear in other asymmetric cases.}\label{fig:evaluation:fig1}
\end{figure}
%\begin{figure}
%  \centering
%  \includegraphics[width=1.0\linewidth]{key-rate-noise-asym.png}
%  \caption{Key-rate bound with CAD when the noise is asymmetric. In particular, $QZ^1 = Q$ while all other Bob's are $QZ^j = \frac{1}{4}Q$.  Here we see CAD can outperform standard QCKA protocol rates.}\label{fig:evaluation:}
%\end{figure}
\begin{figure}
  \centering
  \includegraphics[width=0.75\linewidth]{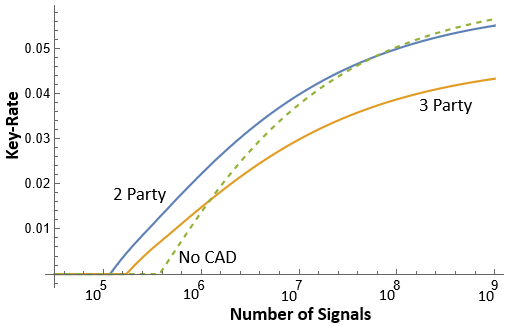}
  \caption{Showing the key-rate versus number of signals for QCKA with CAD (solid) versus without CAD (dashed).  Here, $Q = 10\%$ and we set $QZ^1 = Q$ and all other Bob's to $QZ^j = \frac{1}{4}Q$.  Note, our result can provide positive key-rates for a lower number of signals, than prior work without CAD.}\label{fig:evaluation:fig2}
\end{figure}

%%% Local Variables:
%%% mode: LaTeX
%%% TeX-master: "main-arxiv"
%%% End:

\section{Closing Remarks}%\label{sec:main:}
In this paper, we proposed the addition of CAD to a QCKA protocol introduced in \cite{grasselli2018finite}.
%We derived a novel proof of security for this in the finite key setting, bounding, directly, the quantum min entropy for arbitrary attacks.
%Our method also allowed for asymmetry in the noise between parties.
We evaluated our results and showed that CAD can be beneficial in some scenarios.  In other scenarios, our results are worse. This may be an artifact of our proof technique or it may be due to this particular CAD protocol.  Future work should investigate this further and explore alternative CAD protocols.  %We showed that classical post processing of this nature has the potential to benefit QCKA protocols and our proof techniques may allow future researchers to investigate this area further.
\textbf{Acknowledgments:} WOK would like to acknowledge support from NSF CCF-2143644.

\balance

\end{document}